\documentclass[a4paper,10pt]{article}
\usepackage[utf8x]{inputenc}

\usepackage{graphicx}

\usepackage{amssymb,amsmath,amsthm} 
\usepackage[tight,center]{subfigure} 

\usepackage{amssymb} 
 
\newtheorem{theorem}{Theorem}                                
\newtheorem{proposition}[theorem]{Proposition}

\newtheorem{definition}[theorem]{Definition}

\newcommand{\R}{\mathbb{R}} 
\newcommand{\Rp}{\mathbb{R}_+}

\def \mM {\mathcal S}

\def \mA {\mathcal A} 
\def \mV {\mathcal V} 
 
\def \tM {\widetilde{\mM}} 
\def \ts {\widetilde{S}}

\newcommand{\err}{\Delta} 

\begin{document}

\title{DASS: Detail Aware Sketch-Based\\ Surface Modeling}  

\author{Emilio Vital Brazil\\ 
\textbf{\small IMPA - Instituto Nacional de Matem{\'a}tica Pura e Aplicada}\\
{\small Rio de Janeiro - Brazil}\\ 
\textbf{\small University of Calgary}\\
{\small Calgary - Canada}} 
\date{}

\maketitle

\begin{abstract}
  We present a sketch-based modeling system suitable for detail editing, based on a multilevel representation for surfaces.
  The main advantage of this representation allowing for the control of local (details) and global changes of the model. 
  We used an adaptive mesh (4-8 mesh) and developed a label theory to construct a manifold structure, which is responsible for controlling local editing of the model.
  The overall shape and global modifications are defined by a variational implicit surface (Hermite RBF).
  Our system assembles the manifold structures to allow the user to add details without changing the overall shape, as well as edit the overall shape while repositioning details coherently.
\end{abstract}


\begin{figure}[b]
  \includegraphics[width=\linewidth]{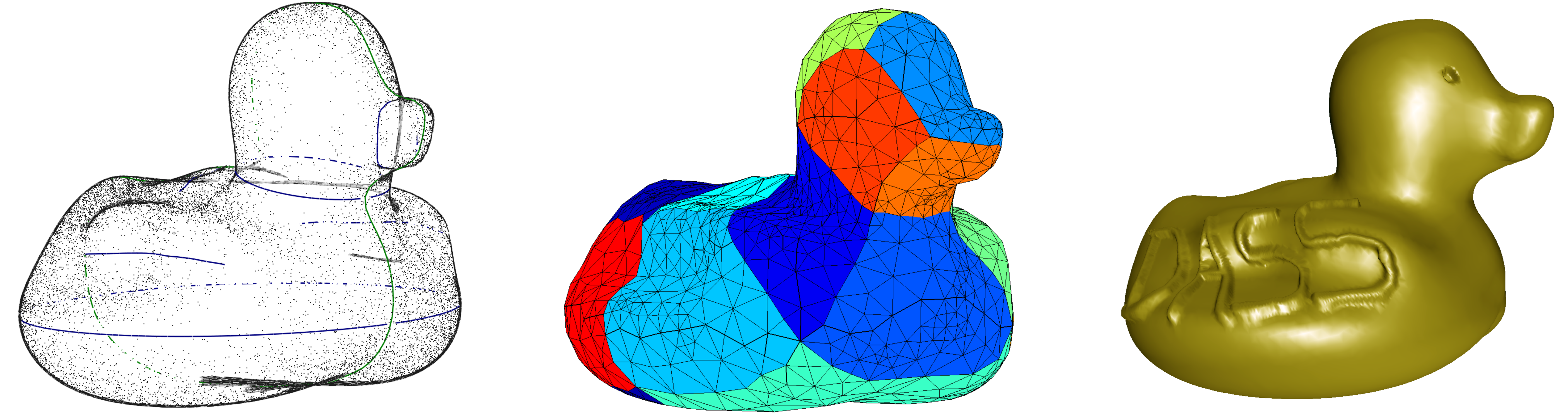} 
 \centering 
    \caption{Modeling a rubber duck: Left to right: implicit model, adaptive mesh, and final model augmented.} 
    \label{fig:composedRep:duck} 
\end{figure}
 
\section{Introduction}
  Sketch-based modeling (SBM) is a solid research area with many interesting problems on different domains such as computer vision, human-computer interaction and artificial intelligence~\cite{olsen09}. 
  However, these applications often do not differentiate the overall shape and local details, encumbering the control of the impact of fine-grained changes to the global shape, and vice versa.
  We thus present an SBM system which is based on a representation conceived with one main objective in mind: good control of global and local transformations using sketch-based tools.
  In Fig.~\ref{fig:composedRep:duck}, for instance, the rubber duck is augmented with local deformations, with nonetheless no changes to the overall shape.
 
   Our pipeline starts with a coarse shape represented by an implicit surface. 
   Specifically, we use Hermite Radial Basis Function (HRBF) due to its support for a great variety of SBM operators, as well as its good projection properties~\cite{vitalBrazil10:sbim}.
   After this, we construct a manifold structure for the implicit surface, which allows us to handle different parameters of the models (such as local augmentation, level of detail, color, and others). 
   We use an adaptive mesh to obtain good frequency control and maintain coherence between global and local transformations.

\section{Pipeline}\label{sec:composedRep:Pipeline} 

  Our pipeline divided in four fundamental parts. 
  We start with the coarse form defined by an implicit surface \cite{vitalBrazil10:sbim}; after that we build a base mesh (Section~\ref{sec:composedRep:BaseMesh}) that has the same topology and approximately the same geometry of the implicit surface. 
  The base mesh induces an atlas (Section~\ref{sec:composedRep:Atlas}) and provides a 4-8 base mesh  (Section~\ref{sec:composedRep:48mesh}).  
  The atlas is built using a partition of the set of faces of the mesh, and we use it to edit the model locally.  
  The 4-8 mesh has two roles in the pipeline: to build a map between surface and atlas, and to visualize the final surface. 
  After we have all parts, the 4-8 mesh is used to edit details that are saved in the atlas, and the atlas maps details onto the 4-8 mesh. 
  In Fig.~\ref{fig:composedRep:pipeline} we depict our pipeline.  
 
   \begin{figure*}[ht!] 
    \centering 
    \includegraphics[width=.55\linewidth]{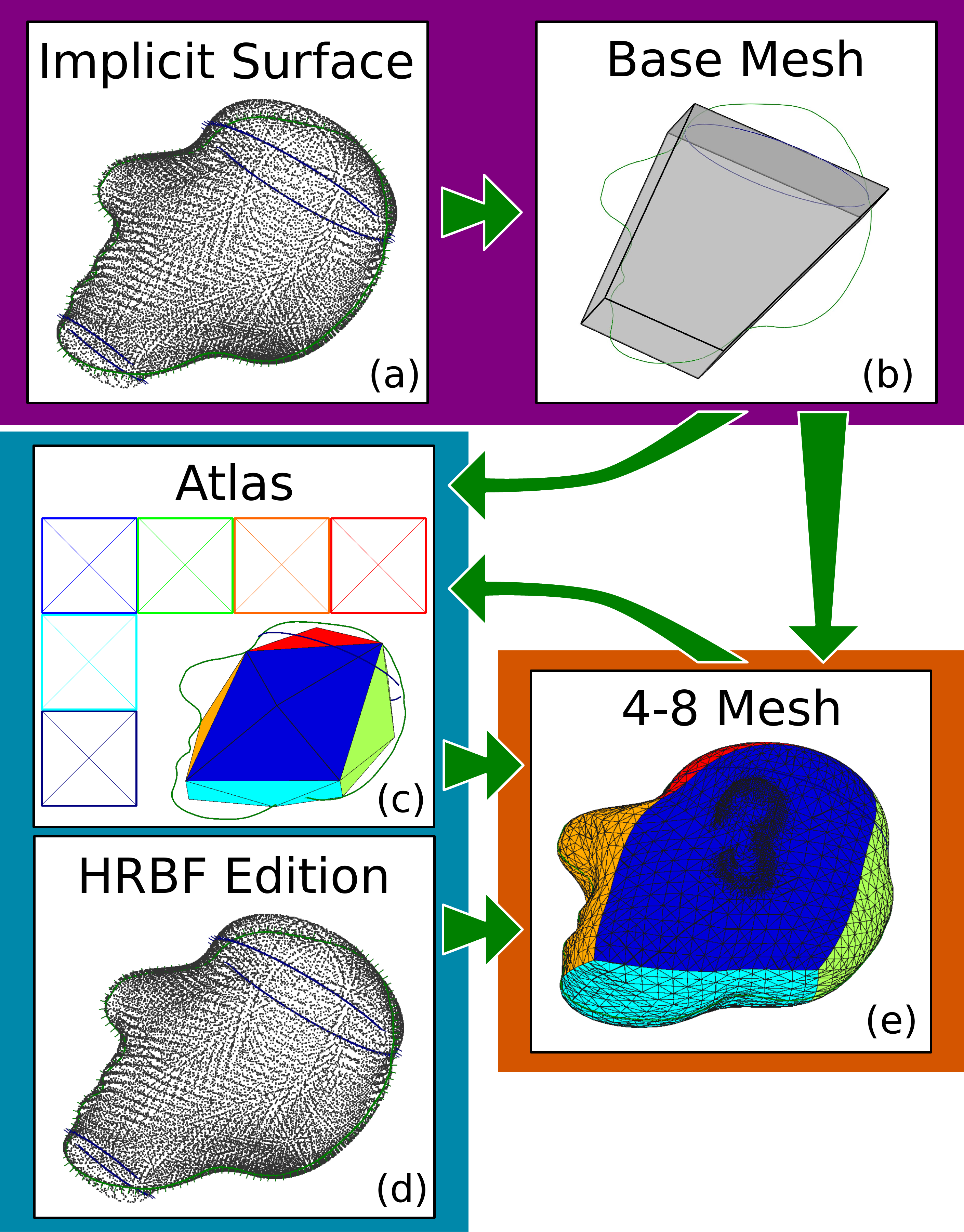} 
     \caption{The pipeline of our Sketch-based surface system. The arrows depict the  information flow.} 
     \label{fig:composedRep:pipeline} 
  \end{figure*} 
 
  The first step in the pipeline is obtaining a coarse shape of the final model (Fig.~\ref{fig:composedRep:pipeline}(a)). 
  Implicit models provide a compact, flexible, mathematically precise representation which are well suited to describe coarse shapes. 
  We use the same implementation described in \cite{vitalBrazil10:sbim},  
  in which the authors introduce a new representation for implicit surfaces and show how it can be used to support a collection of free-form modeling operations. 
  The implicit representation, variational Hermite Radial Basis Function (HRBF) used by Vital Brazil et al. 
  \cite{vitalBrazil10:sbim}, fits well with our pipeline due to its good projection properties, as well as for its simplicity and compactness. 
   
  After we obtain our implicit surface  $\mM$,  we create the manifold structure to represent our final model $S$. 
  To handle parameters, we use an atlas $\mA$ of $S$, i.e., $\mA = \{\Omega_i,\phi_i\}_{i=0}^k$ such that $\Omega_i\subset\R^2$, and $\phi_i:\Omega_i\rightarrow S$ are homeomorphisms~\cite{doCarmo76}. 
  However, we have an implicit surface without information about atlas. 
  There are many approaches to parametrize implicit surfaces, e.g. \cite{bloomenthal94, velho96, stander05}, but in order to find the correct topology of the model these approaches depend on user-specified parameters \cite{bloomenthal94,velho96}, or require differential properties of the surface \cite{stander05}.
  Apart from the topology issue, such methods neither guarantee the mesh quality nor have a direct way to build an atlas structure. 
  As a result, we opted to develop a method that is based on our problem and desired surface characteristics.
 
  First of all, we observe that there are two different scales of detail to be represented: the implicit surface (which is coarse) and the details (which are finer).
  The naive approach would be to use the finest scale of detail to define the mesh resolution.
  However, there are two issues associated with it: firstly, we do not know this finest scale a priori; and secondly, if the details appear in a small area of the model, memory and processing time will be wasted with a heavily refined mesh. 
  To avoid the issues describe in the former paragraph, we adopted a dynamic adaptive mesh, the semi-regular 4-8 mesh \cite{velho04}: it allows for the control of where the mesh is to be fine and coarse, by using a simple error function.
 
  Returning to the problem of parametrization of our implicit surface, now we wish for more than just a mesh: we need an adaptive mesh. 
  The framework presented by \cite{goes08} starts with a semi-regular 4-8 mesh and refines it to approximate surfaces using simple projection and error functions -- from now on we say 4-8 mesh in place of semi-regular 4-8 mesh. 
  To obtain a good approximation of the final surface, the 4-8-base-mesh should have the same topology and should approximate the geometry of the final surface. 
  Thereupon our parametrization problem was reduced to the problems of how to find a good 4-8 base mesh  (Section~\ref{sec:composedRep:BaseMesh}) and how to construct a good error function (Section~\ref{sec:composedRep:48mesh}).   
 
  The parametrization of the implicit surface is built in three parts: base mesh (Fig.~\ref{fig:composedRep:pipeline}(b)), atlas (Fig.~\ref{fig:composedRep:pipeline}(c)), and semi-regular 4-8 mesh (Fig.~\ref{fig:composedRep:pipeline}(e)). 
  In Section~\ref{sec:composedRep:BaseMesh} we present a base mesh with two roles in our system, inducing an atlas for the surface and creating a 4-8 mesh. 
  We describe a method in Section~\ref{sec:composedRep:Atlas} to create an atlas for adaptive meshes based on stellar operators.   
  In Section~\ref{sec:composedRep:48mesh} we discuss how build an error function for the 4-8 mesh that is sensitive to levels of detail. 
 

\section{Base Mesh}\label{sec:composedRep:BaseMesh} 
  The base mesh is the first step to parametrize our surface. This is a very important piece of our pipeline, because three important aspects of the final model depend on the base mesh: the topology of the final model, the atlas, and the 4-8 mesh quality. 
  Finding a good base mesh based only on information from the implicit model is a very hard problem in geometric modeling, but within an approach of a sketch-based modeling proposal, it is natural to make use of user input to obtain more information about the model and create the base-mesh.
 
  The user works with a simple unit of tessellation element (tesel)  which can have the topology of a cube or a torus. 
  This tesel is projected onto the drawing plane, which the user can edit to get a better approximation of the model (geometric and topological) by moving tesel vertices in the plane, dividing a tesel or changing its topology.  
  Afterwards, the system creates a tessellation in the space by moving each tesel vertex along the draw plane's normal direction. 
  In Fig.~\ref{fig:composedRep:base-mesh} we show the typical steps to create the base mesh: the user starts with a bounding box of the sketched lines, then divides tesels, moves vertices and changes tesel topology to build a better approximation of the intended shape.
  Our system defines vertex heights seeking along the draw plane's normal direction, for a root of the implicit surface.  
  Each face defines a chart and then it is triangulated to be the 4-8 base mesh. 
  \begin{figure}[h!] 
    \centering 
    \includegraphics[width=\linewidth]{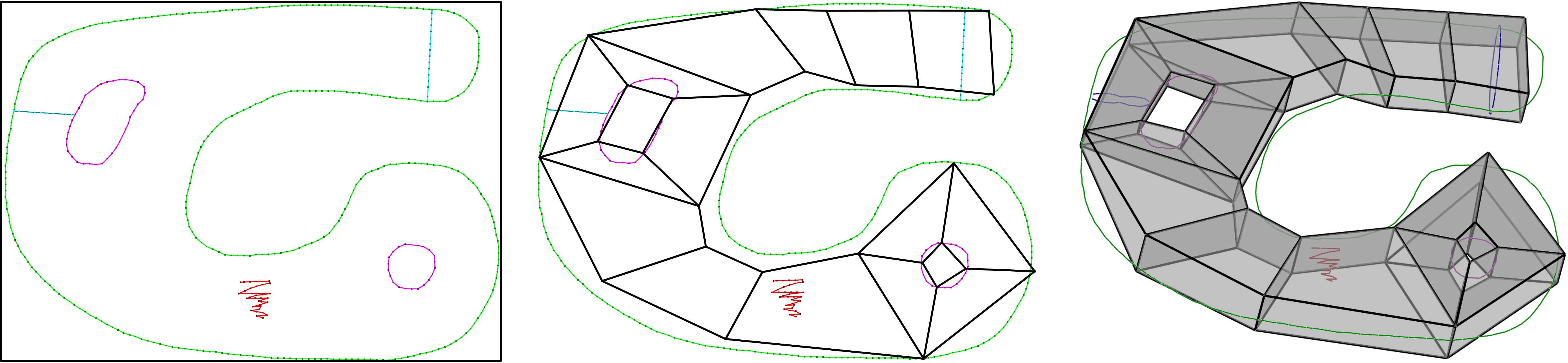} 
    \caption[Creating a base-mesh]{Creating a base-mesh for an implicit surface created using the construct lines described in Vital Brazil et al. \cite{vitalBrazil10:sbim}. Left to right, the first approximate, after the user correct the topology and better approximate the geometry, and the final result in $\R^3$.} 
    \label{fig:composedRep:base-mesh} 
  \end{figure} 

\section{Atlas}\label{sec:composedRep:Atlas} 
  The second step to obtain the manifold structure for our model is to construct an atlas, i.e., a collection of charts $c_i$ that are open sets $\Omega_i\subset\R^2$, and functions $\phi_i:\Omega_i\rightarrow S$ that are homeomorphisms~\cite{doCarmo76}. 
  Specifically for this application, each chart of $\mA$ is associated with a height-map, which is used to define a displacement along the normal direction. 
  In Section~\ref{sec:composedRep:48mesh} we use that height-map to define an error function that locates where the 4-8 mesh need be further refined. 
  We create two types of height-map layers: pre-loaded gray images and height-maps directly sketched on the surface. 
   
  In Fig.~\ref{fig:composedRep:atlas_final} we depict the steps to create an atlas for a 4-8 mesh $M$. 
  After the base mesh is obtained and each of its faces is triangulated, one refinement step is performed and then each base mesh face is associated with a chart (Fig.~\ref{fig:composedRep:atlas_final}(a)). 
  When the mesh is refined to better approximate the geometry, the atlas is updated and the user can draw a curve over the $M$ which is transported to the charts and then this curve creates/edits the height-maps (Fig.~\ref{fig:composedRep:atlas_final}(b)). 
  If the mesh resolution is not enough to represent the details, $M$ is refined; usually that happens when the user creates/edits the height-maps (Fig.~\ref{fig:composedRep:atlas_final}(c)). 
 
   Following we present the main aspects of a \emph{vertex-map} that we developed to allow us to create the atlas structure. 
     \begin{figure*}[ht!] 
     \centering 
     \includegraphics[width=1\linewidth]{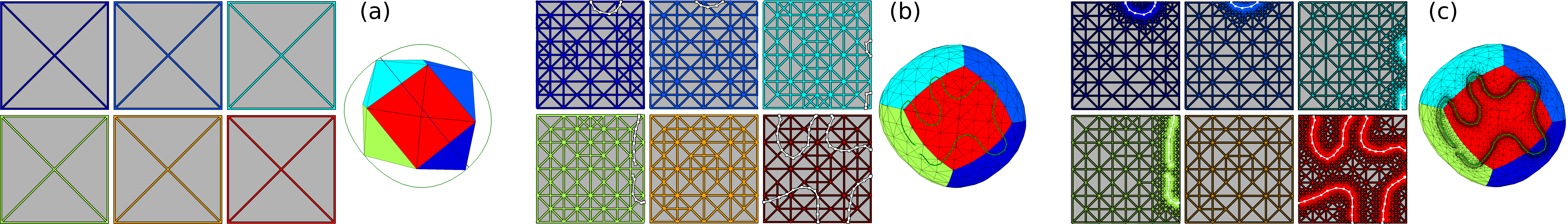} 
     \caption[Atlas steps]{Atlas steps: (a) The atlas is defined after one refinement step of $M$. (b) $M$ is refined and the user defines an augmentation sketching over the surface, and the sketches are transported to $\mA$ to built a height-map. (c) $M$ is refined to represent details of the final surface with height-map.} 
     \label{fig:composedRep:atlas_final} 
  \end{figure*}  
\subsection{Vertex-Map}
   In this section and in the next we construct the theoretical framework to build an atlas using a label function over the vertices of a mesh.
   We work with a general description of adaptive surfaces, based on \emph{stellar subdivision grammars}~\cite{velho03}.
   Our choice for parametric representation, the 4-8 mesh developed by Velho~\cite{velho04}, is an example of application of this grammar. 
   The atlas defined using vertices of the mesh has the following advantages: it is compact and simple; it naturally classifies edges as inner and boundary; and it is suitable for work with dynamic adaptive meshes. 
   
    As discussed before, we need an adaptive mesh to represent the high-frequency details. 
    However, when we do one refinement step in a mesh, new elements (vertices, edges, faces) are created, hence we need to update the atlas. 
    We propose a solution to construct and update an atlas using the natural structure of adaptive surfaces, using a simple label scheme for $4-8$ mesh. 
    Each vertex is labeled as inner vertex of a specific chart or as a boundary; that means if we have $N$ charts there are $N+1$ possible labels.
    The $4-8$ mesh uses stellar operators (Fig.~\ref{fig:composedRep:stellarOperator}), subsequently, we developed rules to update the atlas when these operators are used.
      
   First of all we formalize the concept of the \emph{regular labeled mesh}.
   After that we use these definitions to build an atlas with guarantees for adaptive surfaces that uses Stellar subdivision operators. 
 
  \begin{definition}\label{def:composedRep:klabeled}
    A mesh $M = (V,E,F)$ is \emph{$k$-labeled} if each vertex $v \in V$ has a label $L(v)\in\{0,1,2\dots,k\}$, i.e., if there is $L:V\rightarrow\{0,1,2\dots,k\}$. $L$ is called $k$-\emph{label function}. 
    If $L(v) = i \neq 0$, then $v$ is an \emph{inner-vertex} of the chart $c_i$; if $i=0$, $v$ is a \emph{boundary-vertex}.
  \end{definition}

  \begin{definition}\label{def:composedRep:regularklabeled}
    A face $f\in M$, is \emph{regular $k$-labeled} or \emph{$rk$-face} if there is $v\in f$ with $L(v) \neq 0$ and $\forall\; v_1,v_2 \in f$ such that $L(v_1)\neq0 \neq L(v_2) \Rightarrow L(v_1)=L(v_2)$. 
    A mesh is \emph{regular $k$-labeled} (or \emph{$rk$-mesh}) when all their faces are $rk$-faces.
    The function $L:V\rightarrow\{0,1,2\dots,k\}$ that produces a $rk$-mesh is called a \emph{regular $k$-label} or \emph{$rk$-label}.
  \end{definition}

  Observe that an edge in a regular $k$-labeled mesh has vertices with the same label or one of them has label $0$.
  If the edge has at least one vertex $v$ such that $L(v)=i\neq 0$; we call it an \emph{inner-edge} of the chart $c_i$ or $L(e)=i$; if it has the two vertices labeled as zero it is a \emph{boundary-edge} or $L(e)=0$. 

  \begin{proposition}\label{prop:composedRep:klabelPartition}
    A regular $k$-label function induces a partition on the set of faces. 
  \end{proposition}
  
  \begin{proof}
    Let $M = (V,E,F)$ be a $rk$-mesh. Define the set $F_i = \{ f \in F |\; \exists\; v\in f \mbox{ such that }  L(v) = i \}$, $i\in\{1,2,\dots,k\}$. 
    By definition~\ref{def:composedRep:regularklabeled} every $f\in F$ has at least one $v$ with $L(v)\neq 0$ then:
      $$
	\bigcup_{i=1}^k F_i = F,
      $$
    and if there is more than one $v \in f$ such that $L(v)\neq 0$ then all such vertices will have the same value of $L$, i.e., the face belongs to only one $F_i$, so we conclude:
    $$
      F_i\cap F_j = \varnothing\;\, \hbox{ if }\;i\neq j.
    $$
  \end{proof}

  This proposition allows us to define a collection of charts over a $rk$-meshes.
  We say that a face $f$ is in the chart $c_i$ ($L(f) = i$) if there is at least one $v\in f$ such that $L(v) = i$.
  However for our application it is not enough to have a static map because our mesh is adaptive. 
  Hence we need rules to assign a $L$  value to the new vertices created by the refinement step of the mesh.  

  We develop our techniques based on the two works by Velho \cite{velho03,velho04} on stellar operators and dynamic adaptive meshes.
  We are working with 4-8 mesh which is a multi-resolution triangle-mesh for manifold surfaces. 
  The  4-8 mesh uses to refine and simplify the mesh the stellar operators that come from the theory of the stellar subdivision \cite{lickorish99}.
  Hence we study how to update the atlas after we apply one of these operators: \emph{edge split}, \emph{face split}, and their inverse \emph{edge weld} and \emph{face weld} (Fig.~\ref{fig:composedRep:stellarOperator}). 
  We use the concepts of sequence of meshes $(M_0,M_1,\dots,M_k)$ and \emph{level} of a mesh element exactly as presented by Velho~\cite{velho03}.
  
  \begin{figure}[thbp!]
     \centering
     \includegraphics[width=0.8\linewidth]{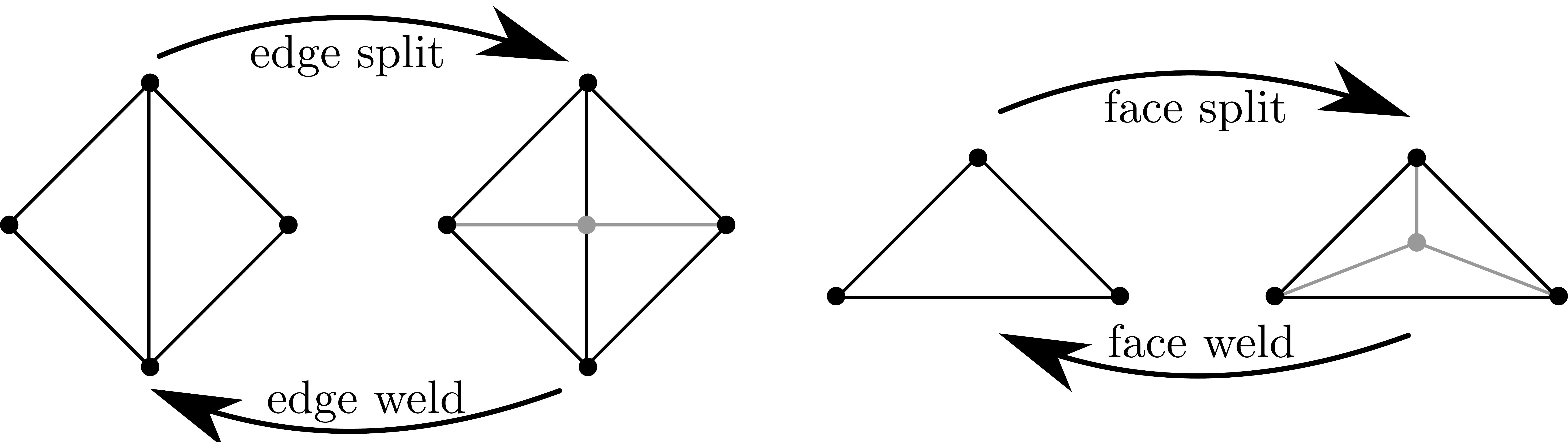}
     \caption[Stellar subdivision operators and inverses.]{Stellar subdivision operators and inverses.}
     \label{fig:composedRep:stellarOperator}
  \end{figure}

  When we apply one of these two stellar subdivision operators (split), it adds only one vertex.
  As a result, to update the atlas we only need rules to label the new vertex $v_n$ created for these two operators over a $rk$-mesh.
    \begin{itemize}
      \item {Face Split --} when the face $f$ is split we define:
	\begin{equation} \label{eq:label_rule_face}
	  L(v_n) = L(f)
	\end{equation} 
      \item {Edge Split --} when the edge $e$ is split we define:
	\begin{equation} \label{eq:label_rule_edge}
	   L(v_n) = L(e)
	\end{equation}
    \end{itemize}
 
  \begin{proposition}
     A stellar subdivision step using the previous rules on a $rk$-mesh~$M$ produces~$M'$ that is a $rk$-mesh too.
  \end{proposition}
  \begin{proof}
      Case we split a face $f$ we create 3 news faces ($f_1, f_2, f_3$), since $M$ is a $rk$-mesh the equation~\eqref{eq:label_rule_face} is well defined. Moreover $v_n \in f_1\cap f_2 \cap f_3$ then they have at least $v_n$ with $L(v_n) = i \neq 0$.
      Since $f$ is a $rk$-face all $v \in f$, $L(v)$ is $0$ or $i$, and for $j=\{1,2,3\}$, $v \in f_j \Leftrightarrow v = v_n$ or $v \in f$, we conclude if $v \in f_j \Rightarrow L(v) = 0$ or $L(v) = i$, i.e, $f_j$ is a $rk$-face.

      The  edge split create four new faces $f_j, j =1,2,3,4$.
      Note that the operator edge split subdivides two faces; lets name these faces \emph{west-face} ($f^w$) and \emph{east-face}  ($f^e$); and their opposite vertex as $v_e$ and $v_w$ respectively. i.e., $v_*\in f^*$ and $v_*\not \in e$.

      If $e$ is an inner-edge then for at least one of its vertices $L(v)=i \neq 0 $.
      Since $e$ is in $f^w$ and $f^e$ we have $L(f^w) = L(f^e) = i$ it implies that if $v\in f^w\cup f^e$ then $L(v) = i$ or $L(v) = 0$. 
      As a result when we split a inner-edge we have $L(v_n)=i$ and $v_n\in \bigcap_jf_j$ and $v\in f_j \Rightarrow v\in f^w\cup f^e$ or $v = v_n$, then $f_j$ is a $rk$-face.
      
      The fact of $e$ being a boundary-edge and $f^w$ and $f^e$ be $rk$-faces imply $L(v_e)\neq 0$ and $L(v_w)\neq 0$.
      Since $v_w \in f_j$ or $v_e \in f_j$ we have at most one $v\in f_j$ such that $L(v)\neq 0$ and $L(v_n)=0$, then we conclude that $f_j$ is $rk$-face.      
  \end{proof}
 
  The simplification step of an adaptive mesh is very important for our application, because when the user changes the sketches the mesh is dynamically updated that implies that the two steps (refinement and simplification) are done. 
  If we start with a $rk$-mesh (level 0) and perform $n$ refinement steps for any $m\leq n$ steps of simplification we yet have a $rk$-mesh.
  This fact is easy to see because when we do a refinement step we do not change the value of the vertices of the current level $j$, thus when we do the inverse operator to simplify only vertices of level $j+1$ are deleted so then the $L$ function over faces is well defined in level $j$. 

  To create a $rk$-mesh using our base-mesh, i.e., to create the $M_0$,  we label all vertices of the base-mesh as boundary ($L(v)=0$) and split each face, the new vertex added is labeled with a new value not $0$.
  After that each face of the base-mesh generates a new chart into the atlas, i.e., if the base mesh had $k$ faces the atlas has $k$ charts.
  In Fig.~\ref{fig:composedRep:M0}  we illustrate the process of creating a mesh $M_0$ that is a $r2$-mesh and three refinement steps.
  
  \begin{figure}[pbth!]
     \centering
     \includegraphics[width=\linewidth]{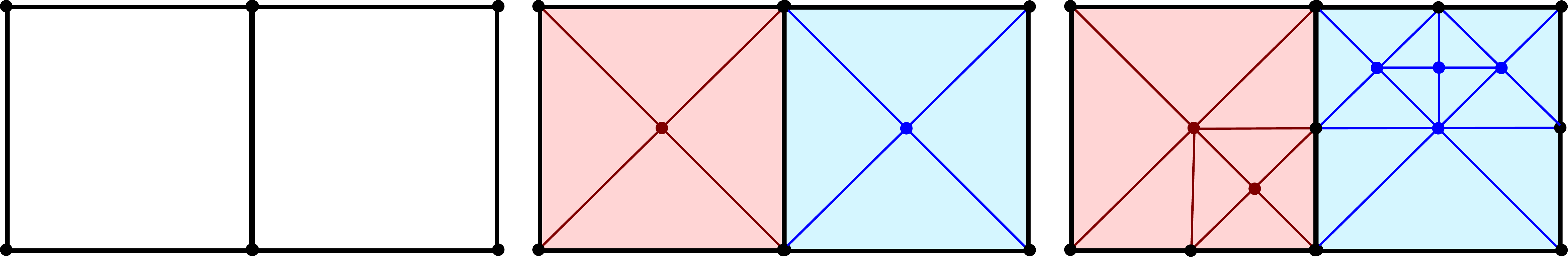}
     \caption[Creating a $rk$-mesh and refinements]{Creating a $r2$-mesh and refinements. Left to right: the base-mesh, $M_0$ which is $r2$-mesh, and after 3 refinement steps: $M_3$. Black elements are boundary ($L(\cdot) = 0$), blue elements are into chart $c_1$ ($L(\cdot) = 1$), and red elements are into chart $c_2$ ($L(\cdot) = 2$). }
     \label{fig:composedRep:M0}
  \end{figure} 

\subsection{Creating a Manifold Structure}
  Now we have a partition over the surface and we know how to refine and simplify the mesh respecting this partition.
  However we do not have yet all elements of an atlas for the surface $S$, is missing to define open set $\Omega_i \varsupsetneq c_i$ and homeomorphisms~$\phi_i$.  
  
  First of all let us to build $c_i$. 
  We define the functions $\phi_i:\Omega_i\big|_{c_i}\rightarrow S$ based on the structure of adaptive levels of the 4-8 mesh $M$.
  And, $c_i = [0,1]\times[0,1]$, then we set the four vertices of the base-mesh face $f_i = \{v_1, v_2, v_3, v_4\}$ to be the boundary of  $c_i$, i.e., the local coordinates in $\Omega_i$ of these vertices are: $v_1^i = (0,0)$, $v_2^i = (1,0)$, $v_3^i = (1,1)$, $v_4^i = (0,1)$.
  We overload the notation for chart, a chart of the atlas, $c_i\in\mA$, has two meanings, the first one is a set of faces, edges and vertices, used in previous section.
  The second one is the parametric space $[0,1]^2\subset\Omega_i$, more precisely when we say a point of $M$ belongs to a chart $c_i$ it means if we can write this points in $\Omega_i$ coordinates, then its coordinates are in $[0,1]^2$.
  At this point all vertices $v$ of $M$ have at least two geometrical information, its coordinates in $\R^3$ and, its coordinates in at least one $\Omega_i$.
  The notation $v^i$ is used to be clear when we are using $v$ in coordinates of $\Omega_i$, how to recover this information we will discuss later.

  Since $M$ is an adaptive mesh and now it has two geometrical aspects, its coordinates in $\R^3$ and in $\mA$, we need rules to update this information.
  When we split an edge $e = \{v_1,v_2\}$ we get its middle  point $v_m$ and project it on $S$ and if $e\in c_i$ then $v^i_m = (v^i_1+v^i_2)/2$.  
  Despite the rules to be simple they achieve goods results, in Section~\ref{sec:composedRep:48mesh} we discuss more about that rules. 
  Now we suppose that the approximation of the adaptive mesh is less than $\varepsilon>0$, i.e., for all points $p$ on $M$ imply $|p-\Pi_S(p)|\leq\varepsilon$ where $\Pi_S(p)$ is the projection of $p$ on the surface $S$.
  We are assuming that $\Pi_S$ is well defined for $\mV_{\varepsilon} = \bigcup_{p\in S}{\mathbf B}(p,\varepsilon)$, where $\mathbf B$ is the open ball with center $p$ and radius~$\varepsilon$. 
  That is true when $\mV_{\varepsilon}\subseteq \mV$, the tubular neighborhood of $S$ \cite{velho96b}.
  Particularly, the vertices of the base mesh start close to $S$ then their projections are well defined, therefore we replace their start position $v$ by $\Pi_S(v)$.
  We will also use the $\Pi_M(p)$, the projection of $p \in S$ on $M$, and again we are supposing that the mesh approximates well the surface.
  We say the chart $c_i$ is well defined after one refinement step (Fig.~\ref{fig:composedRep:M0}), thereupon if a point $p^i\in c_i$ then there is a face $f^i=\{v^i_{1}, v^i_{2} , v^i_{3}\}$  such that $p^i$ is a convex combination of its vertices.
  More precisely  $p^i = \sum_{k=1}^3\alpha_k v^i_{k}$ with $\alpha_k > 0$, $\sum_{k=1}^3\alpha_k=1$.
  So then we define:
  \begin{equation*}
    \phi_i(p^i) = \Pi_S\left( \sum_{k=1}^3\alpha_k \phi_i(v^i_{k})\right).
  \end{equation*}
  Specifically when we split an edge $e$, which belongs to $c_i$, $e^i=\{v^i_1,v^i_2\}$ we have:
  \begin{equation}\label{eq:composedRep:splitProj}
    \phi_i(v^i_n) = \Pi_S\left( \frac{\phi_i(v^i_1)+\phi_i(v^i_2)}{2}\right).
  \end{equation}
  
  \begin{proposition}\label{prop:composedRep:phiEqphi}
     For all $i,j$ and $v\in V$ such that $v \in c_i$ and $v \in c_j$ holds $\phi_i(v^i) = \phi_j(v^j)$. 
  \end{proposition}
  \begin{proof}
     We proof that proposition by induction for all levels of refinement of $M$.
     When we start the charts $c_i$ and $c_j$ all edges that are in their boundary come directly of the base mesh, if $v \in c_i$ and $v \in c_j$ then  $\phi_i(v^i) = \Pi_S(v) = \phi_j(v^j)$, by construction.
     Now suppose the Proposition~\ref{prop:composedRep:phiEqphi} is true for all $v$ with level less or equal the current level. 
     Observe that by~\eqref{eq:label_rule_face} and~\eqref{eq:label_rule_edge} a boundary-vertex $v$ is created only when a boundary-edge is split, consequently by~\eqref{eq:composedRep:splitProj} and induction hypothesis holds:
     \begin{equation*}
       \begin{split}
        \phi_i(v^i) =& \Pi_S\left( \frac{\phi_i(v^i_1)+\phi_i(v^i_2)}{2}\right)\\ 
                    =& \Pi_S\left( \frac{\phi_j(v^j_1)+\phi_j(v^j_2)}{2}\right) = \phi_j(v^j).
       \end{split}
       \end{equation*}
  \end{proof}

  To define the inverse of $\phi_i$ we use the projection $\Pi_M$, the idea is to project the point on the mesh, identify which face it is projected and use the barycentric coordinates to define it coordinates in $\Omega_i$.
  More precisely, let $\Pi_M(p) = \sum_{k=1}^3\alpha_kv_k$, with $\alpha_k > 0$, $\sum_{k=1}^3\alpha_k=1$ and $f = \{v_1,v_2,v_3\}$ where $\nobreak{L(f)=i}$, then we have:
  \begin{equation} \label{eq:composedRep:phiInv}
    \phi^{-1}_i(p) = \sum_{k=1}^3\alpha_kv^i_j.
  \end{equation}
  Since we are supposing that $M$ is close to $S$ we have $\phi$ and $\phi^{-1}$ well defined, i.e., $\phi_i\circ\phi^{-1}_i(p)=p$ and $\phi^{-1}_i\circ\phi_i(p^i)=p^i$ for all $p\in S\cap \phi_i(c_i)$ and $p^i\in c_i$.
  
  To build and to glue the height-maps consistently we need to know how to write inner-points of $c_i$ in $\Omega_j$ coordinates when $c_i$ and $c_j$ are neighbors, i.e., we need be able to write a point $p^i \in c_i$ in $\Omega_j$ coordinates when $c_i$ and $c_j$ have common vertices.
  Since we started our chart with quadrangle domains we use the approach develop by Stam~\cite{stam03} to  convert $p^i$ to $p^j$.
  The author recovers the relative affine coordinates of $\Omega_i$ to $\Omega_j$, he achieves that by matching commons edges of $c_i$ and $c_j$. 
  
\subsection{Sketching over the Surface} 
  To allow the user add an augmentation we freeze the camera and she or he draws polygonal curves over the surface.
  These strokes are transported to atlas $\mA$ where they are used to define height-maps, we name these projected curves as \emph{height-curves}.
  To transport the curves to $\mA$ we use the Equation~\eqref{eq:composedRep:phiInv} in their points, i.e. we project the curve points directly on $M$, identifying which face they was project, and use their barycentric coordinates to transport them to the correspondent $c_i$.
  If the line segment $pq$ starts in the chart $c_i$ and ends in the chart $c_j$ then to guarantee  continuity we write  $p^iq^i$ and find the point of this segment that is in the boundary of $c^i$ and add this point to the height-curve.
  We do the same thing to the segment $p^jq^j$. 
  In Fig.~\ref{fig:composedRep:atlas_lines2} we show the result of this process in two charts.

  \begin{figure}[ht]
    \centering
    \includegraphics[width=\linewidth]{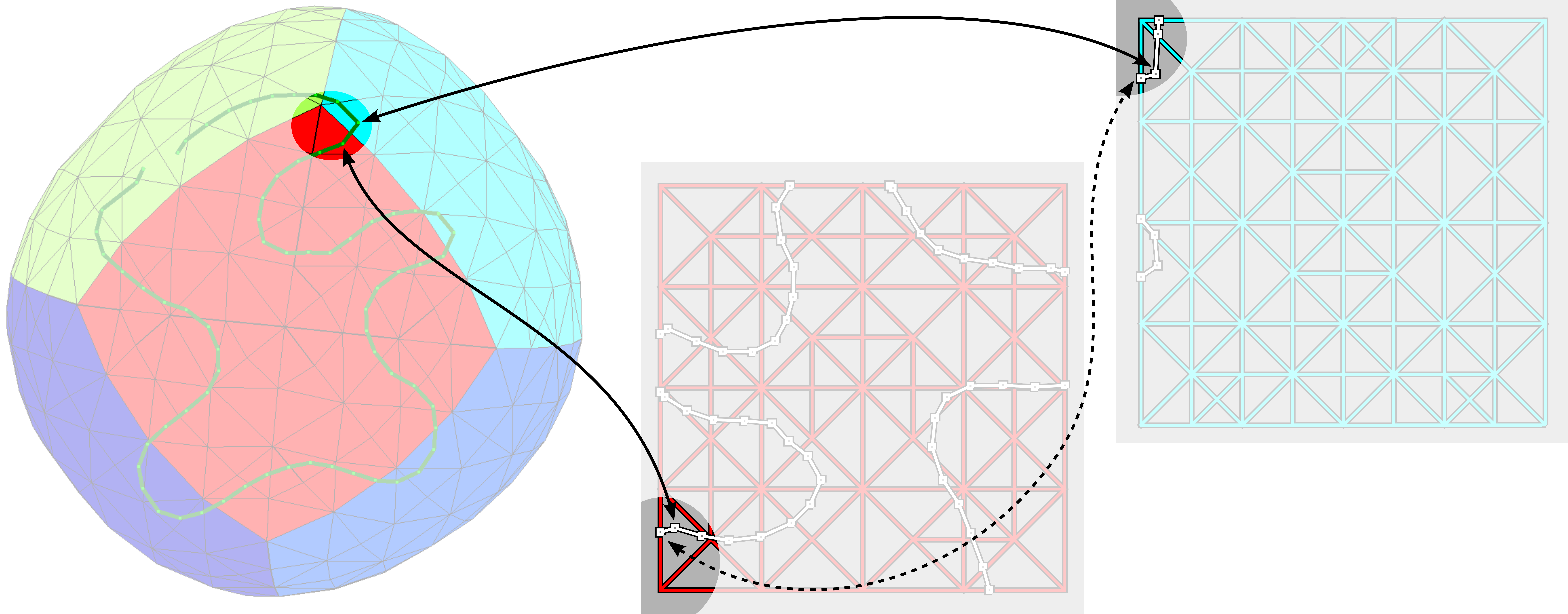}
    \caption[Sketch over surface and transported curve to $\mA$]{Sketch over surface and transported curve to $\mA$. The two solid arrows show points on $M$ that are transported to $\mA$, the dashed arrow shows points that are created in the chart boundaries to guarantee hight-curve continuity.}
    \label{fig:composedRep:atlas_lines2}
  \end{figure} 

  To define the height map we use the distance of a point in $c_i$ to the height-curve.
  We create this field using the approach presented by Frisken~\cite{frisken08}.
  In this work she uses a vector distance field which represent the distance at any point as a vector value.
  For simplicity we define a height $h$ for all points of one same height-curve $\lambda$, and we tested as height function $f_\lambda(p^i)=h\exp(-25d^4/4r^4)$ where $d$ is the distance of $p^i$ to $\lambda$ and $r$ is the radius of influence of $\lambda$.

  After all, we have a height-map~$h^i_u$ for each chart  $c_i$ that was sketched by the user.
  We can compose this height-map with another, such as a gray depth image $h^i_d$, for example to obtain a final height at $p\in M$ adding the heights, $h_p = h^i_d(p^i) + h^i_u(p^i)$.
  Then we have $D(p) = h_p N_p$ where $N_p$ is it normal at $p$. 
  Thus we complete the formulation of the final surface: $\tM = \mM + D(\mM)$ or specifically for all $p\in M$ we have $\tilde{p} = p+ h_pN_p$.

\section{4-8 Mesh}\label{sec:composedRep:48mesh}
  The 4-8 mesh $M$ has two main roles in our system, the first one was described in the last section, we use the mesh to transport points to the atlas.
  Besides that we use $M$ to visualize the approximate final surface.
  We use the library developed by Velho~\cite{velho04}.
  We start the 4-8 mesh with the base mesh triangulated and then we apply one refinement step to define the atlas.
  In addition we need provide a function that samples a edge returning a new vertex, and two error functions.
  One to classify the edges for the refinement step and one to classify the vertices for simplification step.
  
  To define a new vertex we adopt the naive approach that takes the middle point of an edge and project it on surface, i.e., to split a edge $e = \{v_1,v_2\}$, we create a new vertex $v_n = \Pi_S\left((v_1+v_2)/2\right)$ and as described in the last section if $v_n \in c_i$ it saves its local coordinates too.
  In spite of this approach being simple it achieves good results for our application. 
  
  To complete the adaptive process of 4-8 mesh we need to choose which edges will be split, in order to refine the mesh and which vertices will be removed to simplify the mesh.
  In our implementation, this classification is done providing two error functions and one parameter.
  To define our error function we need to describe how we measure the distance between a point and the surface.
  First, observe that $\Pi_S$ is the projection on~$S\neq\ts$, thereupon the $\Pi_S$ is not enough to define the distance.
  To project a point $p$ on $\ts$, first we project $p$ on $S$ then we apply $D$, more precisely, 
  \begin{equation}\label{eq:composedRep:tsProjection}
    \Pi_{\ts}(p)=\Pi_S(p)\oplus D(\Pi_S(p)),     
  \end{equation}
  thus the distance between $p$ to $\ts$ is the usual 
  \begin{equation}\label{eq:composedRep:tsdistance}
    d_{\ts}(p) = |p-\Pi_{\ts}(p)|.    
  \end{equation}

  Now we can determine the error functions using the stochastic approach presented by \cite{goes08}.
  Let us define the error on faces, we randomly take $n$ points on the face and calculate the distance of the point to the surface then we sum all distance and divide by $n$.
  Therefore the error function for edge is the error average  of its faces, and the vertex error function is the error average of its star neighborhood.
  To control the mesh adaptation we define an error threshold $\varepsilon>0$, if the edge error is above that threshold the edge is refined.
  Observe that, the  $\varepsilon$  controls the size of our final mesh.
  If the $\varepsilon$ is  small we have a good approximation of the surface though the mesh will have too many vertices which will be computationally expensive to execute simple operations such as project a line (Fig.~\ref{fig:composedRep:48Error}(c)).
  On the other hand, if the $\varepsilon$ is big the mesh will be  computationally cheap however the mesh will not represent the final surface details (Fig.~\ref{fig:composedRep:48Error}(b)).

  \begin{figure}[ht!]
    \centering
    \subfigure[Detail sketch, $\varepsilon =  10^{-3}$.]{\includegraphics[width=.44\linewidth]{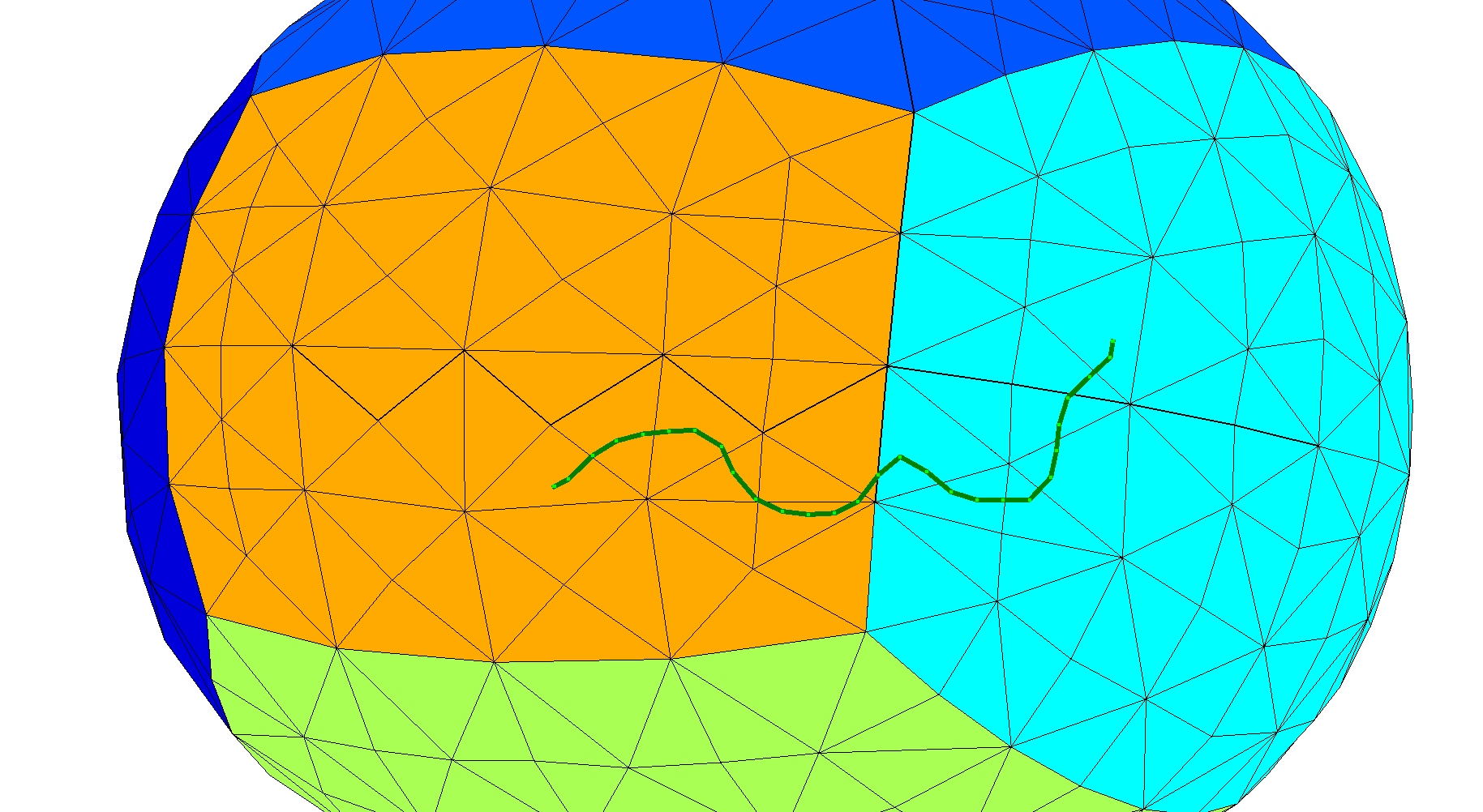}}\hspace{.05\textwidth}
    \subfigure[Simple error function, $\varepsilon = 10^{-3}$.]{\includegraphics[width=.44\linewidth]{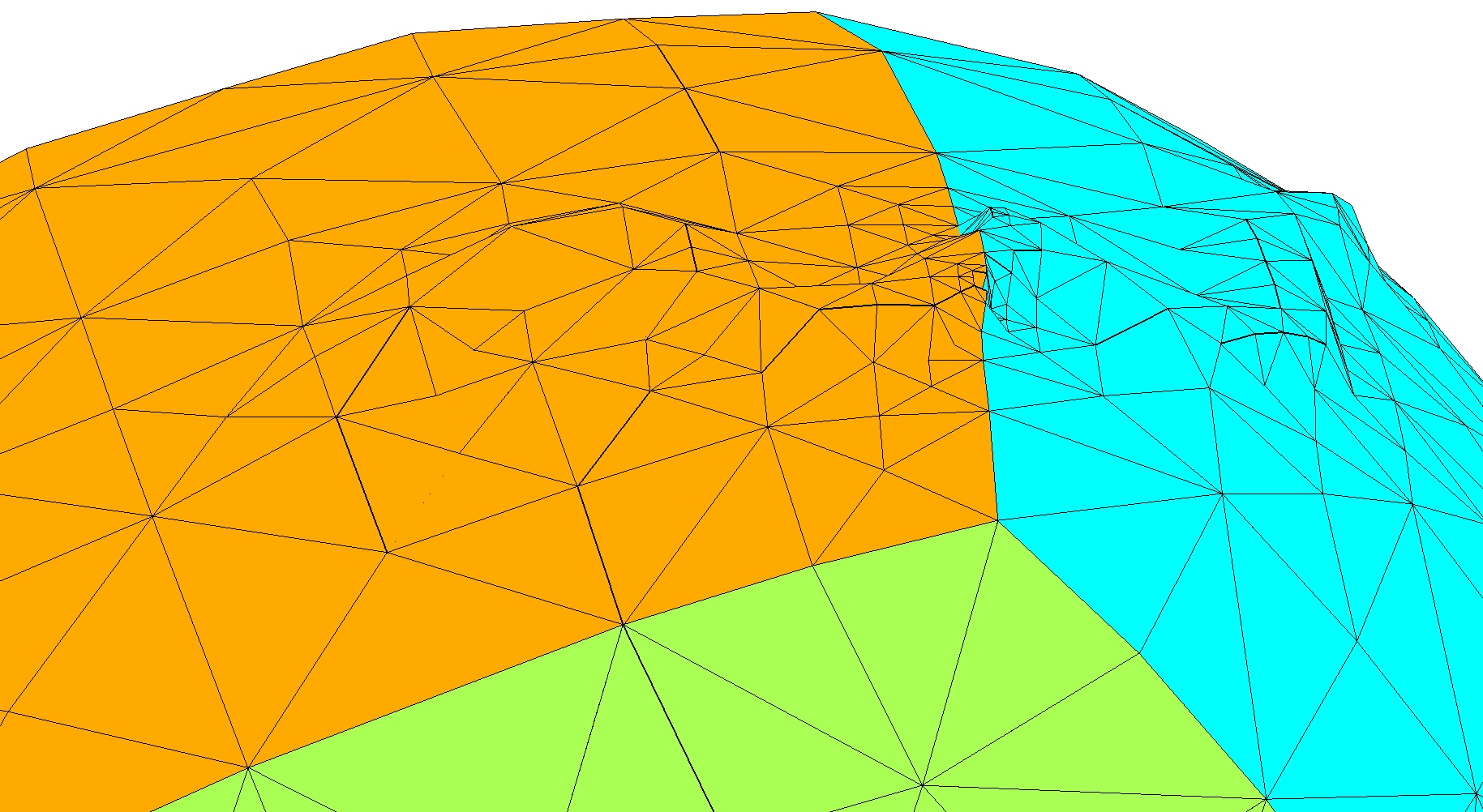}}
    \subfigure[Simple error function, $\varepsilon  = 10^{-4}$.]{\includegraphics[width=.44\linewidth]{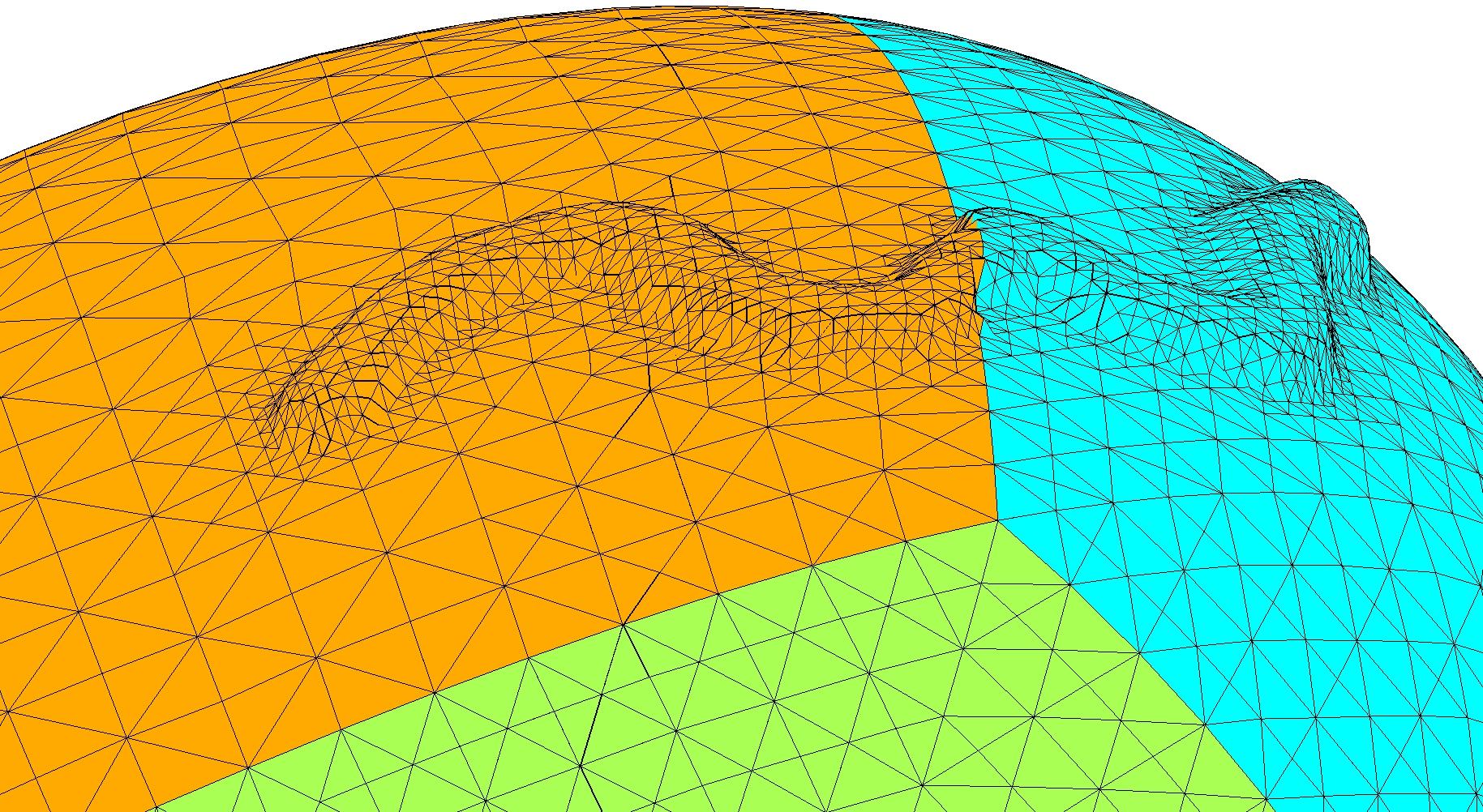}}\hspace{.05\textwidth}
    \subfigure[Local  error function, $\varepsilon =  10^{-3}$.]{\includegraphics[width=.44\linewidth]{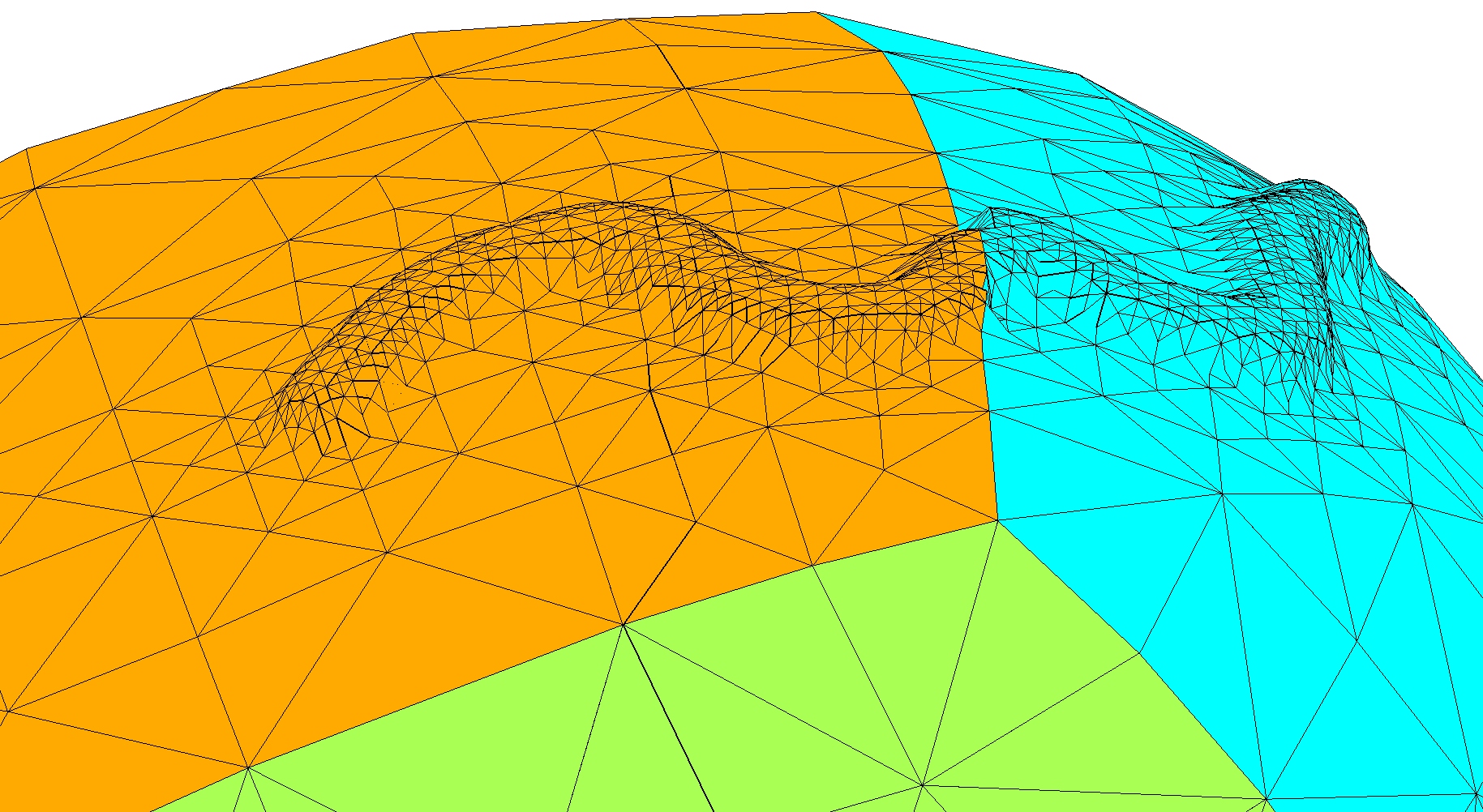}}
   \caption[Local error control]{Local error control.}
    \label{fig:composedRep:48Error}
  \end{figure} 

  It is natural to have an approximation for $S$ geometry coarser than for $\ts$ geometry because we are assuming that $S$ is only the coarse information in contrast to $\ts$ that has also details (Fig.~\ref{fig:composedRep:48Error}(a) and~(c)).
  However, since generally details are restricted to small surface areas if we use $\ts$ to choose~$\varepsilon$ we will have a expensive mesh that do not bring real benefits.
  Since our application works with two different levels of details so then is natural use that structure to define the error functions that depend on the detail level of a surface point.
  In our representation the details are encoded in $D$ however not all parameters of $\mM$ will influence the final mesh, thus we introduce the notation $D_g$ for these parameters that affect the refinement.  
  As a result we define our level of detail at a point $p$ as
  \begin{equation}\label{eq:composedRep:errorP}
   E(p) = \eta( D_g(p)),
  \end{equation}
  where $\eta:\R^d\rightarrow\Rp$. 
  We implement that using the height maps since they are our details over the surface, specifically the Equation~\eqref{eq:composedRep:errorP} is rewritten as $E(p) = \max\{2|\nabla h_p|,1\}$, where $\nabla$ is the gradient.

  Now we have all elements to define an error function that is not blind to level of detail at a point over the surface.
  We define the local error function using  Equation~\eqref{eq:composedRep:tsdistance} and~\eqref{eq:composedRep:errorP} so then we have $\err(p) = d_{\ts}(p)E(p)$.
  Now we apply this new definition in the face error calculation and as result we reformulate the edge error and the vertex error functions.
  In Fig.~\ref{fig:composedRep:48Error} we can observe the difference between to use the simple error function and to use the local error function. The mesh in Fig.~\ref{fig:composedRep:48Error}(b) has 460 vertices however we lost the details of the final surface, if we decrease the~$\varepsilon$ (Fig.~\ref{fig:composedRep:48Error}(c)) we reveal the details though the mesh grows ten times with 4.8k vertices, when we use the local error function (Fig.~\ref{fig:composedRep:48Error}(d)) we reveal the detail and the mesh size does not grow too much, 1.3k vertices.

   \begin{figure*}[t] 
    \centering 
    \includegraphics[width=1 \textwidth]{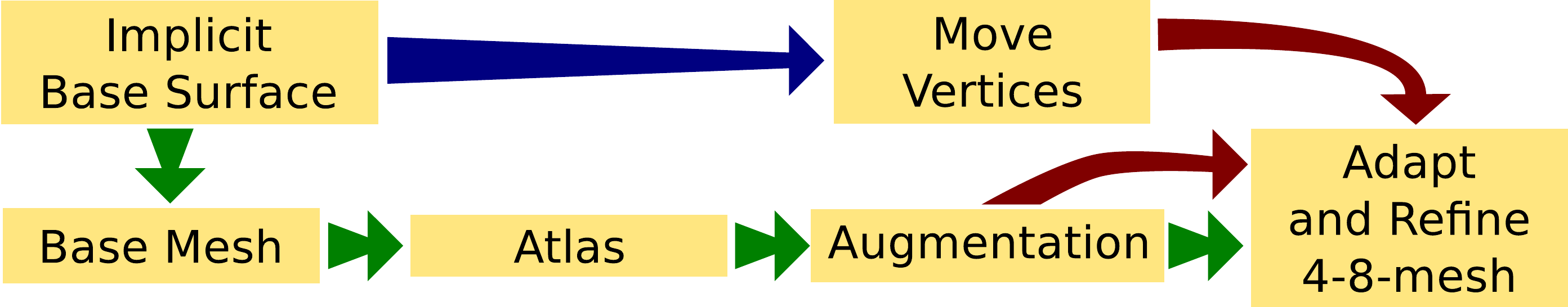} 
    \caption[Overview of system work-flow.]{Overview of system work-flows: green arrows are the startup and topological change step sequence, blue arrow are stepped when the implicit surface is edited, and the red arrow is done when the mesh resolution changes.} 
    \label{fig:composedRep:Overview} 
  \end{figure*} 

\section{Work-flow and Results}\label{sec:composedRep:workflow} 
  In this section we present all pieces of our pipeline working together. 
  Our work-flows are based on the framework presented by \cite{goes08} to adapt dynamic meshes. 
  There are three different work-flows in the pipeline: (1) the user starts the modeling system with a blank page, or by making changes to the actual model topology, (2) the geometry of the implicit surface changes, and (3) the mesh  resolution is recalculated (which usually happens when the height-maps are changed). 
  The overview of the work-flow is depicted in Fig.~\ref{fig:composedRep:Overview}. 
 
   The user starts the model with construct lines, creating samples that define an implicit surface (Fig.~\ref{fig:composedRep:headResult}(a)) using the system described in~\cite{vitalBrazil10:sbim}. 
   After that, the user creates a planar version of the base mesh that approximates the geometry and has the same topology of the final model (Fig.~\ref{fig:composedRep:headResult}(b)). 
   Thus, the base mesh is transported to space (Fig.~\ref{fig:composedRep:headResult}(c)). 
   Now the base mesh is used to create an atlas structure (Fig.~\ref{fig:composedRep:headResult}(d)) for a 4-8 mesh. 
   This mesh is adapted and refined creating the first approximation of the final model (Fig.~\ref{fig:composedRep:headResult}(e)). 
   The steps described up to now are the common steps for all modeling sessions. 
   They are represented by the green arrows in Fig.~\ref{fig:composedRep:Overview}.  
   In addition, these steps are illustrated in Fig.s~\ref{fig:composedRep:spaceCarResult}(a) and~(b), \ref{fig:composedRep:terrainResult}(a), and~\ref{fig:composedRep:visgrafResult}(a). 
   Note that when we change the topology we also need to change the base mesh, restarting the process, e.g., in Fig.~\ref{fig:composedRep:spaceCarResult}(a) and~(b). 
   If there is a predefined height-map, the model reaches the end of this stage with one or more layers of detail.  
   For example, in Fig.~\ref{fig:composedRep:visgrafResult}(a) we start the model with a height-map encoded as a gray image. 
    \begin{figure*}[t] 
    \centering 
    \includegraphics[width=1\textwidth]{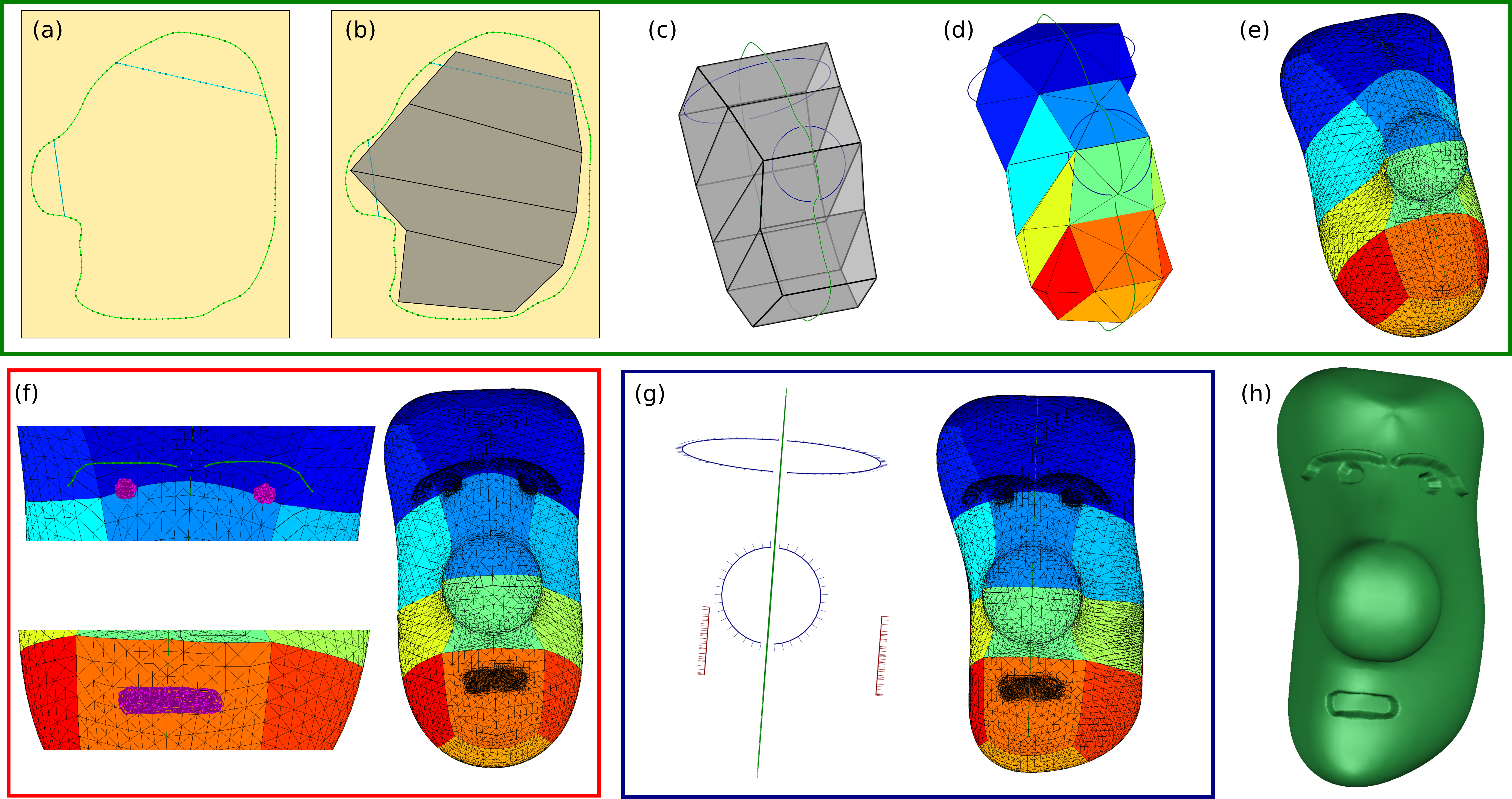} 
    \caption[Steps to model a head.]{Steps to model a head.} 
    \label{fig:composedRep:headResult} 
  \end{figure*} 
 
  \begin{figure*}[t] 
    \centering 
    \includegraphics[width=1\textwidth]{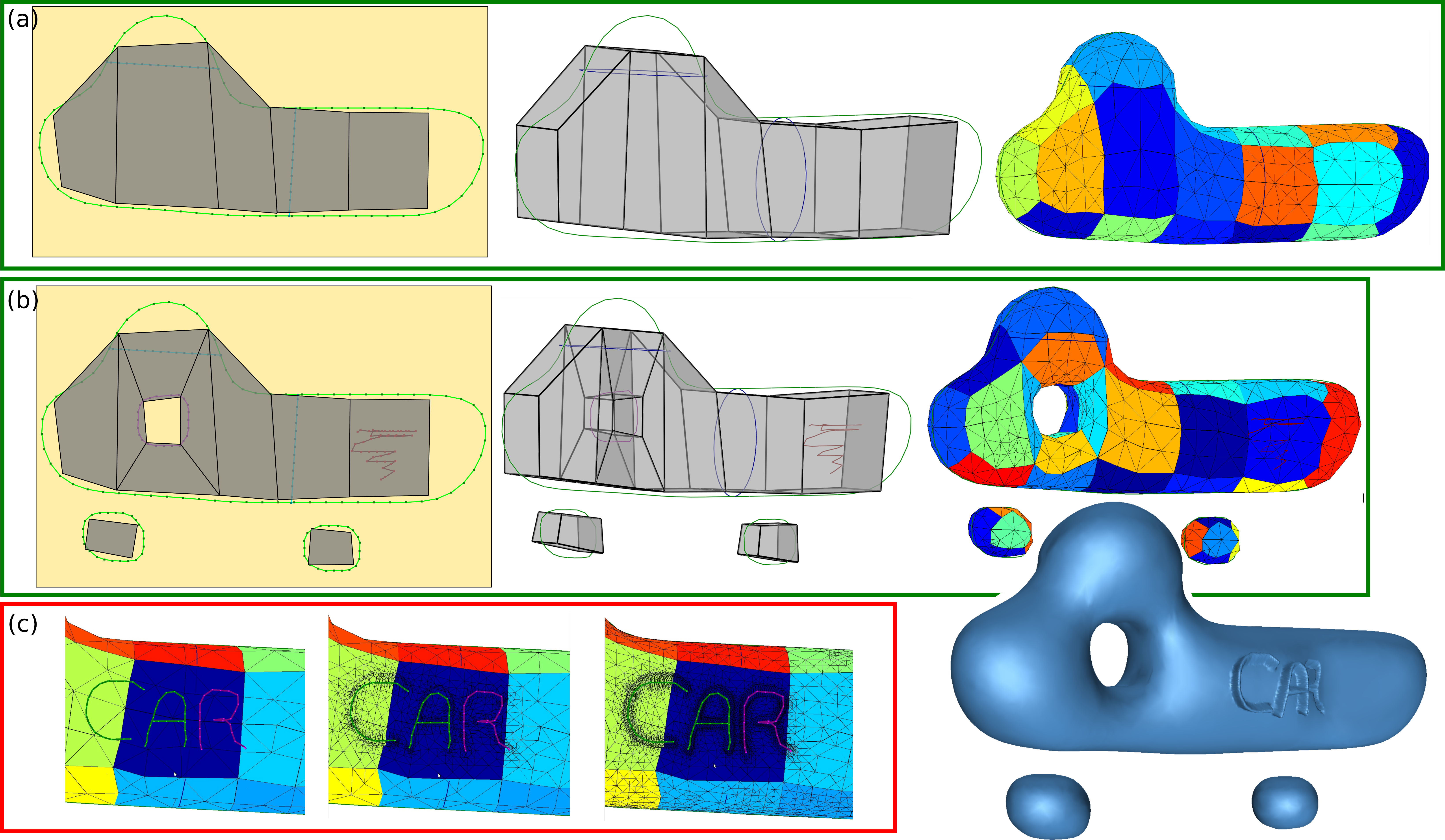} 
    \caption[Steps to model a space car.]{Steps to model a space car.} 
    \label{fig:composedRep:spaceCarResult} 
  \end{figure*} 
 
  After the first approximation for the final surface, the user can edit the implicit surface and create/edit a height-map. 
  When details are added on the surface, in almost all cases it is implied that the resolution of the mesh is not fine enough to represent the new augmentation. In this case we must adapt and refine the mesh. 
  In Fig.s~\ref{fig:composedRep:headResult}(f),~\ref{fig:composedRep:spaceCarResult}(c),~\ref{fig:composedRep:terrainResult}(b), and~\ref{fig:composedRep:visgrafResult}(b): the user sketches a height-map over the surface  and the mesh is refined to represent the geometry of the augmentation correctly. 
  The user can change the implicit surface at any stage, and if the topology is still the same, then the system allows vertices to be moved without adaptation and refinement (in order to obtain a fast approximation).
  Since detail are codified separately, they are moved consistently when implicit surfaces are edited.
  We illustrate that in Fig.s~\ref{fig:composedRep:headResult}(g)~\ref{fig:composedRep:visgrafResult}(c), and~\ref{fig:composedRep:terrainResult}(c),~(e) and~(f). 
  Specifically in  Fig.~\ref{fig:composedRep:terrainResult}(e) and (f) we can compare good final results preserving the details despite the significant changes of the implicit surface.  
  Sometimes, when only the implicit surface is changed, moving the vertices alone is not enough to reach the desired quality. in these cases, the user can adapt and refine the mesh decreasing the error threshold, as shown in Fig.~\ref{fig:composedRep:terrainResult}(d). Here, the user initializes $\varepsilon=10^{-3}$, and after some modeling steps a new threshold of $10^{-4}$ is chosen. 
 
  The modeling of each of the four models presented in this section took approximately 10 minutes, from the blank page stage up to the final mesh generation.
  All the results were generated on an 2.66 GHz Intel Xeon W3520, 12 gigabyte of RAM and OpenGL/nVIDIA GForce GTX 470 graphics. 
  The most expensive step was creating the implicit surface, followed by the creation of the base mesh; on the other hand, processing of the augmentation and minor adjustments in the implicit surface had a minor impact on performance. 
  The bottle neck is the mesh update; if the mesh has too many vertices (around 10k), one refinement step after an augmentation takes about 10 seconds. 
  The final models of space car, terrain, head, and party balloon have 10k, 11k, 11k and 13k vertices respectively.

  \begin{figure*}[ht!] 
    \centering 
    \includegraphics[width=1\textwidth]{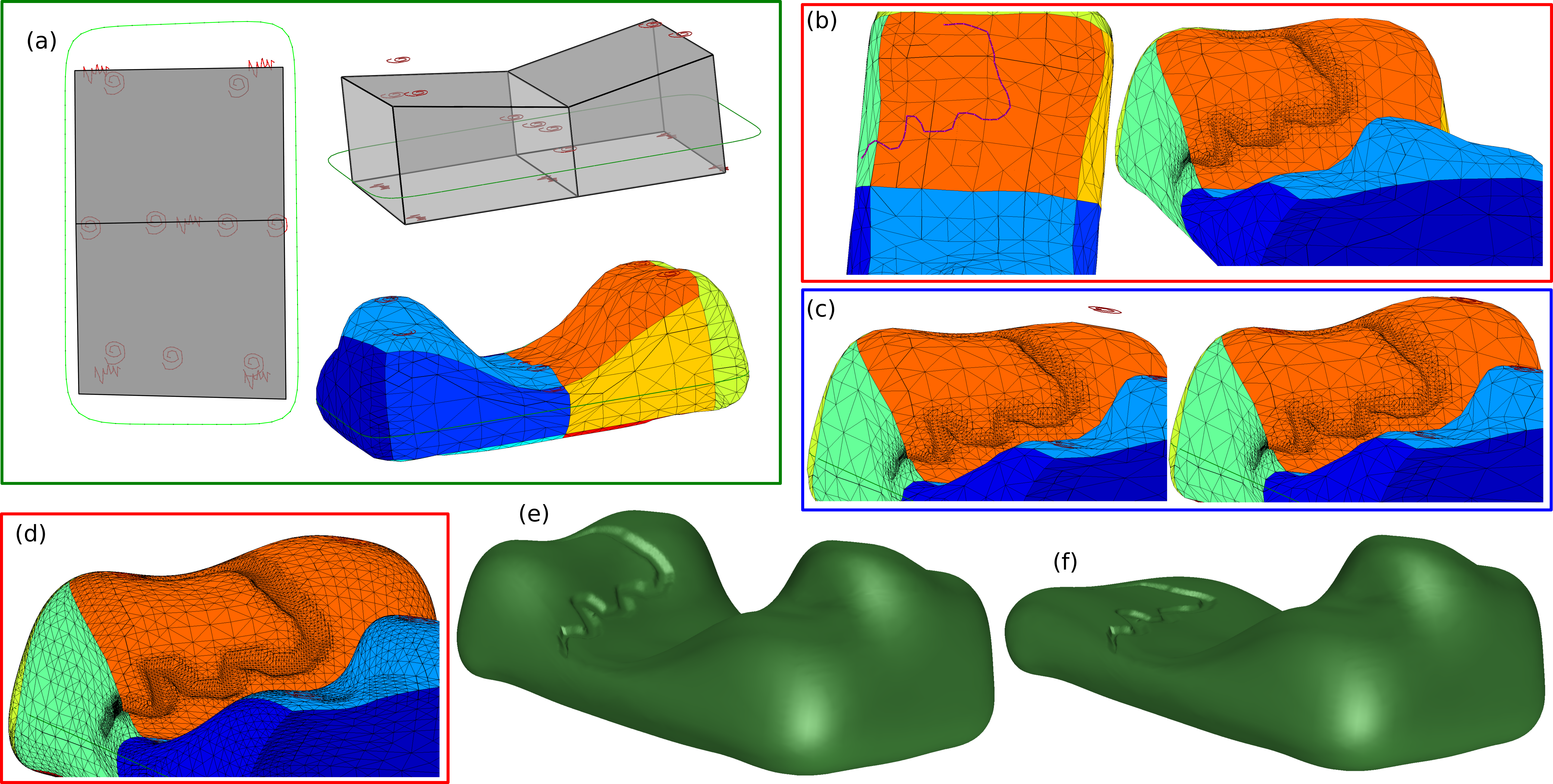} 
    \caption[Steps to model a terrain.]{Steps to model a terrain.} 
    \label{fig:composedRep:terrainResult} 
  \end{figure*} 

  \begin{figure*}[ht!] 
    \centering 
    \includegraphics[width=1\textwidth]{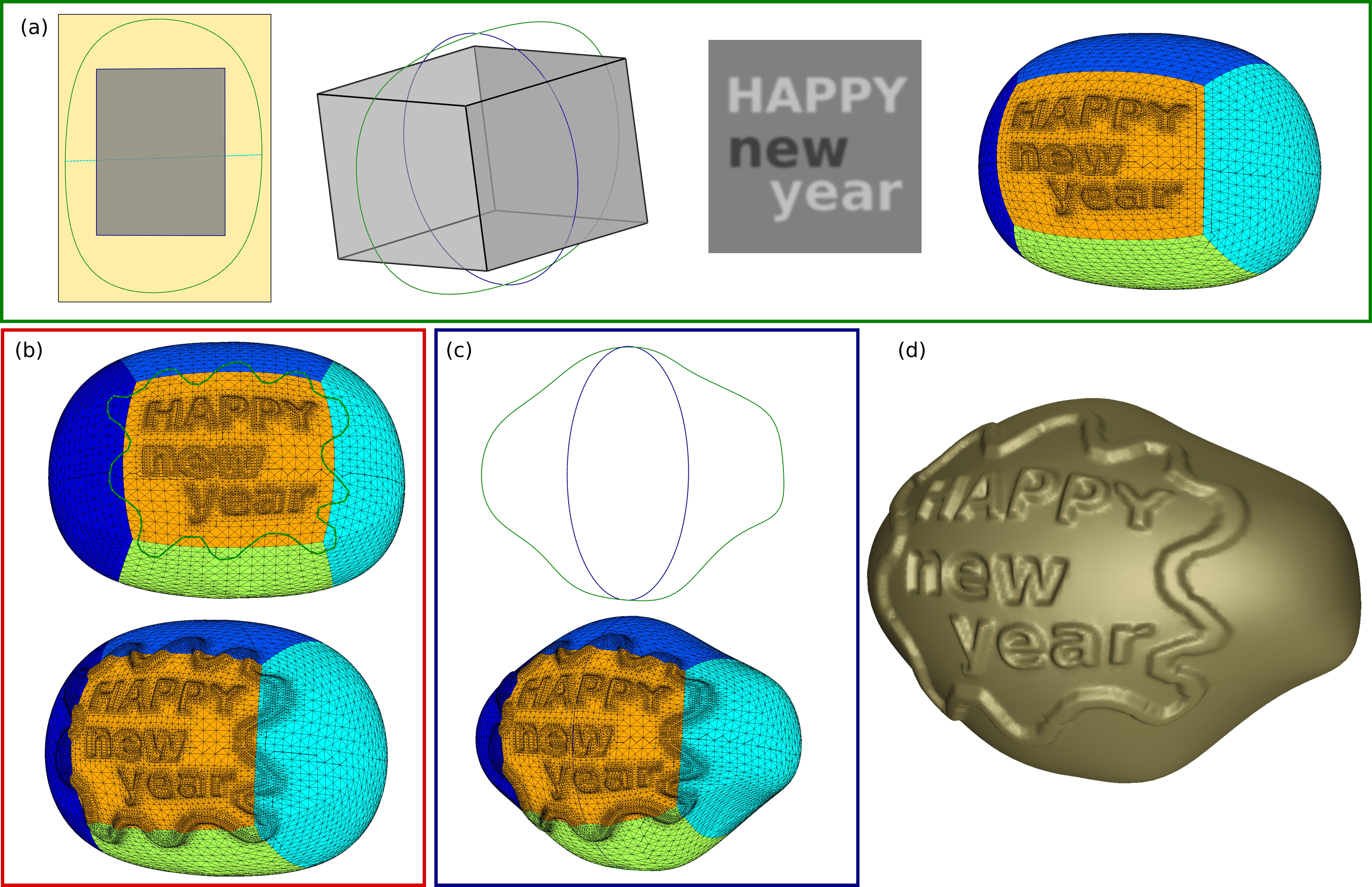} 
    \caption[Steps to model a party balloon.]{Steps to model a party balloon.} 
    \label{fig:composedRep:visgrafResult} 
  \end{figure*}

\section{Conclusion and Future Work}\label{sec:composedRep:conclusion} 
  We developed a Sketch-based Surface Modeling system using a pipeline based on splitting the object between base and detail representations using the semantics of function composition and binary operations, ultimately allowing for control of model editing in two scales: global and local. 
  The proposed pipeline has four main elements: implicit surface, base mesh, atlas and 4-8 mesh.
  Since we desired to investigate whether that representation is powerful enough to build the SBM system, for each pipeline element we either provide an off-the-shelf solution or we create simple approaches for each problem, allowing us to model different shapes controlling the local and global changes. 
  The advantages of these simple approaches are: making sure that the representation is doable and powerful for SBIM problems, and finding avenues for further research.  
 
  One important example of a problem that demands further research is the base mesh. 
  We implement a semi-automatic approach in which the user places the vertices to approximate the geometry and topology, followed by the base mesh creation in the space. 
  This approach achieves good results, however, we only can work in a single plane.  
  Since the base mesh the responsible for the topology of the final model, we are restricted to models which topology can be handled in one plane.  
  Thereupon, we plan to explore two approaches for the base mesh problem. 
  Firstly, we intend to transport the actual semi-automatic solution to 3D, letting the user handle boxes directly in space. The main challenge of this approach is developing an interface in which the user can work efficiently and effectively in creating the base mesh.
  The other approach is using a mesh simplification, as for instance in the method presented by Daniels et al.~\cite{daniels08}. Although this approach is automatic, it starts with a dense mesh; we should then exchange the problem of how to find a base mesh for the problem of creating a mesh with the correct topology. 

  Concerning the Atlas, we aim to develop mathematical and computational tools  to handle the scale of the atlas as well as an interface to control predefined height-maps, and also algorithms that split the atlas if it has a high level of deformation in comparison to the surface. 

  Finally, we want to apply this SBM pipeline for specific domains, since 
  we believe that the potential of our representation and pipeline will be better exploited this way (e.g., figure  modeling or geological modeling). 

\section*{Acknowledgment}
  We would like to thank our colleagues for their useful discussions and advice, in particular to Nicole Sultanum and Ronan Amorin. 
  This research was supported in part by the NSERC / Alberta Innovates Academy (AITF) / Foundation CMG Industrial Research Chair program in Scalable Reservoir Visualization, and grants from the Brazilian funding agencies CNPq and CAPES/PDEE.

\bibliographystyle{IEEEtran}
\bibliography{ARXIV_DASS} 

\begin{thebibliography}{10}
\providecommand{\url}[1]{#1}
\csname url@samestyle\endcsname
\providecommand{\newblock}{\relax}
\providecommand{\bibinfo}[2]{#2}
\providecommand{\BIBentrySTDinterwordspacing}{\spaceskip=0pt\relax}
\providecommand{\BIBentryALTinterwordstretchfactor}{4}
\providecommand{\BIBentryALTinterwordspacing}{\spaceskip=\fontdimen2\font plus
\BIBentryALTinterwordstretchfactor\fontdimen3\font minus
  \fontdimen4\font\relax}
\providecommand{\BIBforeignlanguage}[2]{{%
\expandafter\ifx\csname l@#1\endcsname\relax
\typeout{** WARNING: IEEEtran.bst: No hyphenation pattern has been}%
\typeout{** loaded for the language `#1'. Using the pattern for}%
\typeout{** the default language instead.}%
\else
\language=\csname l@#1\endcsname
\fi
#2}}
\providecommand{\BIBdecl}{\relax}
\BIBdecl

\bibitem{olsen09}
L.~Olsen, F.~F. Samavati, M.~Costa~Sousa, and J.~Jorge, ``Sketch-based
  modeling: a survey,'' \emph{Computer \& Graphics}, vol.~33, no.~1, pp.
  85--103, 2009.

\bibitem{vitalBrazil10:sbim}
E.~{Vital Brazil}, I.~Mac\^edo, M.~{Costa Sousa}, L.~H. {de Figueiredo}, and
  L.~Velho, ``Sketching variational {H}ermite-{RBF} implicits,'' in \emph{SBIM
  '10: 7th Eurographics Workshop on Sketch-Based Interfaces and Modeling},
  2010, pp. 1--8.

\bibitem{doCarmo76}
M.~P. do~Carmo, \emph{Differential geometry of curves and surfaces}.\hskip 1em
  plus 0.5em minus 0.4em\relax Englewood Cliffs, N. J.: Prentice-Hall Inc.,
  1976.

\bibitem{bloomenthal94}
J.~Bloomenthal, \emph{An implicit surface polygonizer}.\hskip 1em plus 0.5em
  minus 0.4em\relax San Diego, CA, USA: Academic Press Professional, Inc.,
  1994, pp. 324--349.

\bibitem{velho96}
L.~Velho, ``Simple and efficient polygonization of implicit surfaces,''
  \emph{Journal of graphics, gpu, and game tools}, vol.~1, no.~2, pp. 5--24,
  February 1996.

\bibitem{stander05}
B.~T. Stander and J.~C. Hart, ``Guaranteeing the topology of an implicit
  surface polygonization for interactive modeling,'' in \emph{SIGGRAPH 2005
  Courses}.\hskip 1em plus 0.5em minus 0.4em\relax ACM, 2005.

\bibitem{velho04}
L.~Velho, ``A dynamic adaptive mesh library based on stellar operators,''
  \emph{Journal of graphics, gpu, and game tools}, vol.~9, no.~2, pp. 21--47,
  2004.

\bibitem{goes08}
F.~de~Goes, S.~Goldenstein, and L.~Velho, ``A simple and flexible framework to
  adapt dynamic meshes,'' \emph{Computers {\&} Graphics}, vol.~32, no.~2, pp.
  141--148, 2008.

\bibitem{velho03}
L.~Velho, ``Stellar subdivision grammars,'' in \emph{Proc. of Eurographics/ACM
  SIGGRAPH symp. on Geom. proc.}, ser. SGP'03.\hskip 1em plus 0.5em minus
  0.4em\relax Eurographics Association, 2003, pp. 188--199.

\bibitem{lickorish99}
W.~B.~R. Lickorish, ``Simplicial moves on complexes and manifolds,'' in
  \emph{Proceedings of the {K}irbyfest ({B}erkeley, {CA}, 1998)}, ser. Geom.
  Topol. Monogr., vol.~2.\hskip 1em plus 0.5em minus 0.4em\relax Geom. Topol.
  Publ., Coventry, 1999, pp. 299--320 (electronic).

\bibitem{velho96b}
L.~Velho and J.~Gomes, ``{Approximate conversion of parametric to implicit
  surfaces},'' in \emph{Computer Graphics Forum}, vol.~15, no.~5, 1996, pp.
  327--337.

\bibitem{stam03}
J.~Stam, ``Flows on surfaces of arbitrary topology,'' \emph{ACM Trans. Graph.},
  vol.~22, pp. 724--731, July 2003.

\bibitem{frisken08}
S.~F. Frisken, ``Efficient curve fitting,'' \emph{Journal of graphics, gpu, and
  game tools}, vol.~13, no.~2, pp. 37--54, 2008.

\bibitem{daniels08}
J.~Daniels, C.~T. Silva, J.~Shepherd, and E.~Cohen, ``Quadrilateral mesh
  simplification,'' \emph{ACM Trans. Graph.}, vol.~27, pp. 148:1--148:9,
  December 2008.

\end{thebibliography}
\end{document}